\documentclass[10pt,journal]{IEEEtran}

\usepackage{amsmath,amssymb,amsthm,mathtools}
\usepackage{graphicx}
\usepackage{booktabs}
\usepackage{tikz}
\usepackage{cite}
\usepackage{url}
\usepackage{microtype}
\usepackage{balance}
\usetikzlibrary{arrows.meta,positioning,calc}

\newtheorem{theorem}{Theorem}
\newtheorem{proposition}{Proposition}

\newtheorem{corollary}{Corollary}

\newcommand{\Ccal}{\mathcal{C}}
\newcommand{\Rcal}{\mathcal{R}}
\newcommand{\Ecal}{\mathcal{E}}
\newcommand{\Pe}{\operatorname{Pe}}
\newcommand{\clconv}{\operatorname{cl\,conv}}
\newcommand{\dd}{\mathrm{d}}
\newcommand{\E}{\mathbb{E}}
\newcommand{\Prob}{\mathbb{P}}

\newcommand{\phin}{\phi}
\newcommand{\Phin}{\Phi}

\title{Fast--Slow Communication with Endogenous Transport}

\author{Lav R. Varshney%
\thanks{L. R. Varshney is with the AI Innovation Institute and the Department of Electrical and Computer Engineering, Stony Brook University, Stony Brook, NY 11794 USA (e-mail: lav.varshney@stonybrook.edu).}\thanks{ChatGPT 5.6-Sol was used to support implementation of computations, writing, and drawing figures.}}

\begin{document}
\maketitle

\begin{abstract}
A communication system may convey urgent information through a fast physical stream and more specific information through a slower material stream. In several biological and engineered settings, however, the fast process also changes the transport law of the slow one. We study this architecture under a shared resource constraint, with a strictly increasing concave fast-channel capacity--cost function and a deadline-constrained slow molecular channel. We first characterize the capacity region under separated message routing and message-independent operating-point schedules, and identify the marginal criterion for complementarity rather than competition between the streams. For one-dimensional drift diffusion, we prove that arrival probability before a deadline is strictly log-concave in P\'eclet number. For a distinguishable-token deadline-erasure channel, any increasing concave transport-actuation law then yields an exact single-crossing theorem: complementarity exists if and only if an initial transport-assistance elasticity exceeds one, the transition is unique when it exists, and the decreasing allocation branch remains the Pareto boundary after convexification. For positive baseline drift and sufficiently strong coupling, a unique critical normalized deadline determines when complementarity disappears. Short- and long-deadline limits clarify the associated temporal regimes. Numerical examples for a finite-frame LTI-Poisson slow channel exhibit analogous allocation behavior with counting noise and intersymbol interference.
\end{abstract}

\section{Introduction}
\label{sec:intro}
Communication need not take place on a single physical time scale. A local event may launch a rapidly propagating pressure, electrical, or mechanical perturbation while also releasing slower chemical carriers. The fast stream can provide an early alert and then the slower stream can carry a more detailed message. This fast--slow communication architecture is especially intriguing when the two information streams are not just parallel \cite{GargSLZLWK2005}, but when the first process also changes the physical channel of the second.   

Wound-induced signaling in plants has exactly this structure.  Rather than electric potentials \cite{awan2019plant}, a wound can produce rapid hydromechanical pressure changes while poroelastic relaxation drives water flow that transports chemical elicitors through the xylem \cite{bacheva2025hydromechanical}. Mechanical perturbations can be decoded through mechanosensitive pathways, whereas molecular elicitors can engage ligand-gated or receptor-mediated pathways. The hydromechanical response is generic enough that it is unlikely, by itself, to distinguish the full range of upstream events and downstream responses. Yet it helps create the mass flow by which a chemically-specific signal travels. Thus the fast stream is simultaneously informative and transport-enabling.

Molecular communication has been studied through diffusion, advection, first-passage timing, receptor dynamics, counting noise, and channel memory \cite{berger2003living,gohari2016directions,pierobon2010physical,pierobon2013capacity, NakanoEH2024}. Positive drift leads to the additive inverse-Gaussian timing channel \cite{srinivas2012aign}; diffusion-induced memory leads naturally to LTI-Poisson models \cite{aminian2015lti}; and finite particle lifetimes create explicit useful-reception deadlines \cite{farsad2018finite}. There is also a growing body of work in which an electric field is used to accelerate molecular transport or reshape the channel impulse response \cite{ma2019electric,chou2022timevarying,cho2022electrophoretic,chou2026field}, establishing that transport assistance can improve capacity. Here, the assisting physical mechanism itself supports a fast information stream and may compete with molecular signaling for a common resource.

At a more abstract level, action-dependent-state channels allow an encoder action both to convey information and to influence a subsequent channel state \cite{weissman2010action}. The model here pursues a message-independent operating point that allocates resources between two orthogonal streams; the fast codeword does not control the slow channel in a symbol-by-symbol manner. As such, we demonstrate a novel capacity--cost phenomenon rather than needing new coding theorems (instead drawing on known coding theorems that establish operational interpretation for diffusion channels with memory \cite{hsieh2013foundations}).

The focus is asking when there starts to be a resource tradeoff between the fast stream and the slow stream. A unit of resource shifted toward fast signaling produces more fast information, but it may also improve slow transport enough to offset the molecular payload that was displaced. The resulting allocation curve can initially have positive slope: both rates rise together. Only after transport assistance saturates does  competition between the two streams reappear.  We do assume the resources for fast and slow are commensurable through an extensive quantity such as energy.

This paper investigates the tradeoff in three steps. First, it gives the exact capacity region for separated operating-point schedules and a criterion that applies to any increasing fast-channel capacity--cost function, thereby providing a universal architectural principle. Second, for a drift--diffusion channel with a reception deadline, it establishes structural properties of first-passage reliability that imply a unique complementarity-to-competition transition, an exact existence condition, and a critical normalized deadline. This argument generically shows that under an increasing concave transport law, the decreasing deterministic branch is concave in fast rate and hence remains the Pareto boundary of the full capacity region. Third, a finite-frame LTI-Poisson channel model incorporates stochastic molecule counts and within-frame intersymbol interference. Its numerically-evaluated capacity retains the same short-, intermediate-, and long-deadline behavior.

Note that the model here considers independent fast and slow messages, but the two streams might describe the same event at different levels of urgency and specificity in the motivating plant communication setting. This source coding aspect is especially important in control settings \cite{sahaiMitterUnstable}, but distinct from the channel question studied herein.

\section{Fast-Assisted Slow Communication}

Consider the following model, as depicted in Fig.~\ref{fig:architecture}.  As noted, we restrict to separated message routing: the fast message is encoded only in the fast input and decoded from the fast output, whereas the slow message is encoded only in the slow input and decoded from the slow output. Conditional on the operating point, the two channel laws factor, and their noises are independent. 

\subsection{Signaling episodes and a shared resource}
Communication uses $n$ independent signaling episodes. In each episode, the fast-channel output is available by time $T_f$, whereas the slow-channel output consists of observations collected through a later deadline $T>T_f$. Coding and block-message decoding are performed across episodes. Thus $T_f$ and $T$ constrain per-episode observation times rather than the decoding latency of an $n$-episode code.  Rates are measured in bits per episode and multiplication by an episode repetition frequency converts them to bits per unit time.

Independence of molecular episodes requires a physical reset. This can be realized by degradation or active clearance at the deadline, sufficiently long guard intervals, or episode-specific molecular labels.  Without a reset, late molecules create inter-episode memory and hence a different problem.

Let $\Gamma>0$ be a normalized resource budget per episode. An operating point allocates $a\in[0,\Gamma]$ to the fast mechanism and $b\ge 0$ to the slow mechanism, with
\begin{equation}
 a+b\le \Gamma.
 \label{eq:budget}
\end{equation}
Specific physical costs would be commensurated through $a$ and $b$. Depending on the application, $\Gamma$ may represent energy or similar extensive quantities that impose constraint.

The split schedule is deterministic, known to both terminals, and independent of the messages. It may vary across subblocks to implement time sharing. Allowing the split itself to carry information would be action-dependent-state coding \cite{weissman2010action}, which we do not consider.

\subsection{Generic fast and slow channels}
Let $W_f(y|x)$ be the fast memoryless channel with input cost $\rho_f(x)$. Its capacity--cost function is
\begin{equation}
 C_f(a)=\sup_{P_X:\,\E[\rho_f(X)]\le a} I(X;Y_f).
 \label{eq:Cf}
\end{equation}
As usual \cite{mceliece2002theory}, $C_f$ is nondecreasing and concave. We assume it is continuous, $C_f(0)=0$, and strictly increasing over the operating range. A mechanosensitive receptor channel, or conventional physical channels may all be represented by \eqref{eq:Cf}.

The resources put into the fast channel also establish a scalar slow-channel transport condition
\begin{equation}
 \nu=\nu(a),\qquad \nu'(a)\ge 0.
 \label{eq:nu}
\end{equation}
The quantity $\nu$ may be drift velocity, field strength, flow, or another monotone transport-quality parameter. Let $W_{s,T}^{(\nu)}$ denote the slow channel observed only through deadline $T$, and let
\begin{equation}
 C_s(\nu,b;T)
 \label{eq:Cs}
\end{equation}
be its capacity under slow-input cost $b$. Assume continuity and monotonicity in $\nu$ and $b$. The time-scale separation is represented by taking $\nu(a)$ to be a deterministic operating condition on the slow time scale; the detailed fast symbols average out rather than dynamically modulate the slow channel.

\begin{figure}
\centering
\begin{tikzpicture}[
 >=Latex,
 every node/.style={font=\footnotesize},
 block/.style={draw,rounded corners,minimum height=7.5mm,minimum width=25mm,align=center}
]
\node[block] (fast) {fast channel\\cost $a$};
\node[left=8mm of fast] (mf) {$M_f$};
\node[right=8mm of fast] (yf) {$Y_f$};
\node[block,below=7mm of fast] (state) {transport state $\nu(a)$};
\node[block,below=7mm of state] (slow) {deadline-truncated\\slow channel};
\node[left=8mm of slow] (ms) {$M_s$};
\node[right=8mm of slow] (ys) {$Y_s$};
\draw[->] (mf) -- (fast);
\draw[->] (fast) -- (yf);
\draw[->] (fast) -- node[right] {$a\mapsto\nu(a)$} (state);
\draw[->] (state) -- (slow);
\draw[->] (ms) -- node[below] {$b=\Gamma-a$} (slow);
\draw[->] (slow) -- (ys);
\end{tikzpicture}
\caption{Fast-assisted slow communication. Resource $a$ supports the fast stream and establishes the transport condition of the deadline-constrained slow stream; the remaining resource $\Gamma-a$ supports the slow payload.}
\label{fig:architecture}
\end{figure}
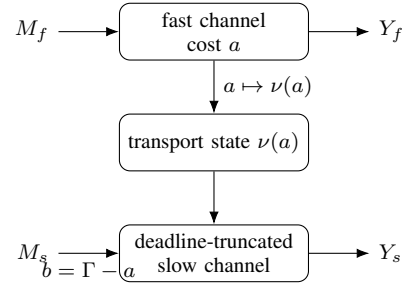

Again, note that as depicted in 
Fig.~\ref{fig:architecture}, $a$ enters twice: through the fast capacity $C_f(a)$ and through the slow transport state $\nu(a)$.
For fixed $a$, monotonicity in $b$ implies that the upper-right rate pair uses $b=\Gamma-a$. Define
\begin{equation}
 F(a)=C_f(a),\quad
 S_T(a)=C_s(\nu(a),\Gamma-a;T).
 \label{eq:FS}
\end{equation}

\begin{proposition}[Capacity region]
Under message-independent operating-point schedules, the fixed-split region is
\begin{equation}
 \Rcal(a)=[0,F(a)]\times[0,S_T(a)].
 \label{eq:rectangle}
\end{equation}
The capacity region is
\begin{equation}
 \Ccal_{\Gamma,T}
 =\clconv\!\left(\bigcup_{0\le a\le\Gamma}\Rcal(a)\right).
 \label{eq:capacityregion}
\end{equation}
\end{proposition}
\begin{IEEEproof}
For fixed $a$, independent fast and slow codes achieve the Cartesian product in \eqref{eq:rectangle}. Conversely, orthogonality of the streams and the fixed known operating point bound each rate by its corresponding capacity--cost function. A deterministic schedule that uses several operating points yields the average of their rate pairs; hence time sharing gives the convex hull. Closure accounts for limiting schedules and boundary rates.
\end{IEEEproof}

This coding statement is direct, but what we aim to understand is the shape of the family of corners $(F(a),S_T(a))$ and which part remains after convexification in \eqref{eq:capacityregion}.

\section{Generic Capacity--Cost Geometry}
Assume for now that $C_f$, $C_s$, and $\nu$ are differentiable. Since $F'(a)=C_f'(a)>0$, shifting resource toward the fast mechanism always increases fast capacity. For the slow stream,
\begin{align}
 S_T'(a)
 &=C_{s,\nu}(\nu(a),\Gamma-a;T)\nu'(a)
   -C_{s,b}(\nu(a),\Gamma-a;T).
 \label{eq:Sprimegeneric}
\end{align}
The first term is the marginal slow-rate benefit created by improved transport. The second term is an opportunity cost of removing the same resource from the molecular input.

Let
\begin{equation}
 \Delta_T(a)=C_{s,\nu}\nu'(a)-C_{s,b},
 \label{eq:Delta}
\end{equation}
where the derivatives are evaluated at $(\nu(a),\Gamma-a;T)$.

\begin{theorem}[Marginal fast--slow criterion]
For every interior allocation,
\begin{equation}
 \frac{\dd S_T}{\dd F}
 =\frac{\Delta_T(a)}{C_f'(a)}.
 \label{eq:slope}
\end{equation}
Thus $\Delta_T(a)>0$ is a complementary regime in which additional fast investment raises both rates, whereas $\Delta_T(a)<0$ is a competitive regime in which fast rate rises and slow rate falls. 
\end{theorem}
\begin{IEEEproof}
Apply the chain rule to \eqref{eq:FS} and use $F'(a)>0$.
\end{IEEEproof}

Notice the fast capacity--cost function determines how quickly the curve moves horizontally in rate space, but it does not determine where complementarity ends and competition for resources starts. That transition is set entirely by slow-channel properties and resource opportunity cost. To obtain a complete characterization rather than just a local criterion, we now specialize the slow transport law.

\section{Deadline-Constrained Drift Diffusion}
We consider a particular class of slow molecular channels.

\subsection{First passage and dimensionless parameters}
Consider one-dimensional molecular motion from a transmitter at $0$ to an absorbing receiver at $L$,
\begin{equation}
 \dd Z_t=v\,\dd t+\sqrt{2D}\,\dd W_t,
 \qquad Z_0=0,
 \label{eq:sde}
\end{equation}
where $D>0$ is molecular diffusivity and $v\ge0$ is drift velocity. The first-passage time
\begin{equation}
 \tau=\inf\{t\ge0:Z_t=L\}
 \label{eq:tau}
\end{equation}
has inverse-Gaussian density \cite{srinivas2012aign}
\begin{equation}
 f_\tau(t|v)
 =\frac{L}{\sqrt{4\pi Dt^3}}
 \exp\!\left[-\frac{(L-vt)^2}{4Dt}\right].
 \label{eq:igpdf}
\end{equation}
The probability of arrival by deadline $T$ is
\begin{align}
 q_T(v)
 &=\Prob\{\tau\le T\}\nonumber\\
 &=\Phin\!\left(\frac{vT-L}{\sqrt{2DT}}\right)
 +e^{vL/D}\Phin\!\left(-\frac{vT+L}{\sqrt{2DT}}\right).
 \label{eq:qv}
\end{align}

Introduce the normalized deadline and P\'eclet number \cite{squires2005microfluidics}
\begin{equation}
 \theta=\frac{DT}{L^2},
 \quad p=\Pe=\frac{vL}{D}.
 \label{eq:dimensionless}
\end{equation}
The ratio $\theta$ compares the allowed response time to the diffusion time $L^2/D$, while $p$ compares directed advection to diffusion. Their product
\begin{equation}
 p\theta=\frac{vT}{L}
 \label{eq:advective_distance}
\end{equation}
measures the source-to-receiver distances traversed by mean advection during the available time.

In dimensionless form,
\begin{align}
 Q(\theta,p)
 &=\Phin(A)+e^p\Phin(B),
 \label{eq:Q}\\
 &A=\frac{p\theta-1}{\sqrt{2\theta}},
 \quad
 B=-\frac{p\theta+1}{\sqrt{2\theta}}.
 \label{eq:AB}
\end{align}
Thus $q_T(v)=Q(\theta,p)$. Define the logarithmic transport sensitivity
\begin{equation}
 \ell_\theta(p)=\frac{\partial}{\partial p}\log Q(\theta,p).
 \label{eq:ell}
\end{equation}

Now consider the geometry needed for a resource allocation theorem.
\begin{theorem}[Geometry of deadline reliability]
For every $\theta>0$ and $p\ge0$:
\begin{enumerate}
\item $Q$ is strictly increasing in $p$, with
\begin{equation}
 Q_p(\theta,p)=e^p\Phin(B)>0;
 \label{eq:Qp}
\end{equation}
\item
\begin{equation}
 0<\ell_\theta(p)\le\tfrac12,
 \label{eq:ellbound}
\end{equation}
with equality if and only if $p=0$;
\item $Q(\theta,p)$ is strictly log-concave in $p$, equivalently
\begin{equation}
 \frac{\partial}{\partial p}\ell_\theta(p)<0; \mbox{ and }
 \label{eq:logconcave}
\end{equation}
\item for every $p>0$, $\ell_\theta(p)$ is strictly decreasing in $\theta$, with
\begin{equation}
 \lim_{\theta\downarrow0}\ell_\theta(p)=\frac12,
 \qquad
 \lim_{\theta\to\infty}\ell_\theta(p)=0.
 \label{eq:elllimits}
\end{equation}
At $p=0$, $\ell_\theta(0)=1/2$ for all $\theta$.
\end{enumerate}
\end{theorem}
\begin{IEEEproof}
See Appendix~\ref{app:reliability}.
\end{IEEEproof}

Strict log-concavity is a key point, showing deadline reliability has diminishing fractional returns to P\'eclet number.

\subsection{Deadline-erasure molecular channel}
For a closed-form capacity law, consider distinguishable molecular tokens. Each token carries a symbol from an alphabet of size $M$. If it reaches the receiver by $T$, its identity is recovered; otherwise an erasure is declared and the token is removed or rendered irrelevant before the next episode. A token opportunity is therefore an $M$-ary erasure channel with success probability $Q(\theta,p)$ and capacity $Q(\theta,p)\log_2M$.

Let $b$ be the normalized number of token opportunities and set $\beta=\log_2M$. The slow capacity is
\begin{equation}
 C_s(p,b;\theta)=\beta b Q(\theta,p).
 \label{eq:erasurecapacity}
\end{equation}
Note the separation between payload and transport reliability. 

Let the P\'eclet number be generated by the fast operating point:
\begin{equation}
 p=p(a),
 \quad p'(a)>0,
 \quad p''(a)\le0.
 \label{eq:concavetransport}
\end{equation}
The concavity assumption represents diminishing transport returns to fast investment. Along the resource boundary,
\begin{equation}
 S_\theta(a)=\beta(\Gamma-a)Q(\theta,p(a)).
 \label{eq:Stheta}
\end{equation}
Define the transport-assistance elasticity
\begin{equation}
 \Ecal_\theta(a)
 = (\Gamma-a)p'(a)\ell_\theta(p(a)).
 \label{eq:elasticity}
\end{equation}
Then
\begin{equation}
 S_\theta'(a)
 =\beta Q(\theta,p(a))\,[\Ecal_\theta(a)-1].
 \label{eq:Sprimeelasticity}
\end{equation}
The elasticity compares the fractional gain in deadline success from one unit of fast resource with the one-for-one loss of slow payload.

\begin{theorem}[Exact transition and Pareto boundary]
Suppose \eqref{eq:concavetransport} holds and $C_f$ is strictly increasing and concave.
\begin{enumerate}
\item $\Ecal_\theta(a)$ is strictly decreasing on $[0,\Gamma)$.
\item A complementary interval exists if and only if
\begin{equation}
 \Gamma p'(0)\ell_\theta(p(0))>1.
 \label{eq:exactcondition}
\end{equation}
If \eqref{eq:exactcondition} does not hold, $S_\theta(a)$ is maximized at $a=0$. If it holds, there is a unique $a_c\in(0,\Gamma)$ satisfying
\begin{equation}
 \Ecal_\theta(a_c)=1,
 \label{eq:ac}
\end{equation}
with $S_\theta$ strictly increasing on $[0,a_c]$ and strictly decreasing on $[a_c,\Gamma]$.
\item Every split $a<a_c$ is Pareto-dominated by $a_c$.
\item On $[a_c,\Gamma]$, $S_\theta$ is strictly concave. Consequently, $S_\theta$ expressed as a function of $F=C_f(a)$ is concave, and the decreasing deterministic branch is the Pareto boundary of the convexified region \eqref{eq:capacityregion}.
\end{enumerate}
\end{theorem}
\begin{IEEEproof}
See Appendix~\ref{app:allocation}.
\end{IEEEproof}

This theorem gives the phase structure beyond just the local derivative test. There can be at most one transition. When it exists, the positive-slope portion is not itself a capacity tradeoff: it is dominated by the point $a_c$. The optimality frontier begins where the slow rate is maximized and then follows the concave decreasing branch.

For an affine transport law,
\begin{equation}
 p(a)=p_0+\eta a,
 \qquad \eta>0,
 \label{eq:affine}
\end{equation}
the exact condition is
\begin{equation}
 \Gamma\eta\,\ell_\theta(p_0)>1.
 \label{eq:affineexact}
\end{equation}
The universal sensitivity bound \eqref{eq:ellbound} immediately yields the following.

\begin{corollary}[Coupling threshold]
For the affine law \eqref{eq:affine},
\begin{equation}
 \Gamma\eta>2
 \label{eq:necessarycoupling}
\end{equation}
 is necessary for a complementary interval at any deadline. If $p_0=0$, it is also sufficient for every finite $\theta$.
\end{corollary}

The strict inequality has a natural interpretation: the maximal fractional transport gain per unit fast resource must be large enough to overcome the immediate payload loss.

\subsection{A critical deadline}
For positive baseline drift, the deadline itself produces a sharp phase boundary.

\begin{theorem}[Critical normalized deadline]
Let $p_0=p(0)>0$ and $\kappa=\Gamma p'(0)$.
\begin{enumerate}
\item If $\kappa\le2$, no complementary interval exists for any $\theta>0$.
\item If $\kappa>2$, there is a unique $\theta_c\in(0,\infty)$ satisfying
\begin{equation}
 \kappa\ell_{\theta_c}(p_0)=1.
 \label{eq:thetac}
\end{equation}
Complementarity exists exactly for $0<\theta<\theta_c$.
\item For $\theta<\theta_c$, the transition allocation $a_c(\theta)$ is strictly decreasing in $\theta$ and
\begin{equation}
 \lim_{\theta\uparrow\theta_c}a_c(\theta)=0.
 \label{eq:aclimit}
\end{equation}
As $\theta\downarrow0$, $a_c(\theta)$ converges to the unique solution of
\begin{equation}
 (\Gamma-a)p'(a)=2.
 \label{eq:shortroot}
\end{equation}
\end{enumerate}
\end{theorem}
\begin{IEEEproof}
The exact existence condition is $\kappa\ell_\theta(p_0)>1$. Theorem~2 shows that, for $p_0>0$, $\ell_\theta(p_0)$ decreases continuously from $1/2$ to $0$, proving the first two claims. The strict decrease of $a_c(\theta)$ follows from the implicit function theorem because $\partial_a\Ecal_\theta<0$ and $\partial_\theta\Ecal_\theta<0$. The limits follow from continuity and using \eqref{eq:elllimits}.
\end{IEEEproof}

When $p_0=0$, $\ell_\theta(0)=1/2$ for every deadline. If $\Gamma p'(0)>2$, a positive complementary interval therefore persists for every finite $\theta$. It becomes a boundary layer: $a_c(\theta)\to0$ and its rate gain vanishes as $\theta\to\infty$.

\subsection{Fast--slow regimes}
The theorems can be interpreted in terms of three regimes, as shown in Table~\ref{tab:regimes}.

\emph{Short deadline: fast effectively alone.}
For bounded feasible P\'eclet number,
\begin{equation}
 \sup_{0\le a\le\Gamma}Q(\theta,p(a))\longrightarrow0
 \quad\text{as }\theta\downarrow0.
 \label{eq:shortQ}
\end{equation}
Thus the absolute slow rate vanishes even when its derivative with respect to fast assistance is positive. Early information is almost entirely fast.

\emph{Intermediate deadline: fast unlocks slow.}
When arrival is neither negligible nor saturated, $\ell_\theta(p)$ is large enough that \eqref{eq:exactcondition} can hold. The allocation curve first rises in both coordinates, reaches $a_c$, and only then becomes a tradeoff. In this regime, fast signaling is an enabling resource for the slow stream.

\emph{Long deadline: transport saturation and competition.}
Uniformly on bounded P\'eclet intervals,
\begin{equation}
 Q(\theta,p)\longrightarrow1
 \quad\text{as }\theta\to\infty.
 \label{eq:longQ}
\end{equation}
Hence
\begin{equation}
 S_\theta(a)\longrightarrow\beta(\Gamma-a)
 \label{eq:longS}
\end{equation}
uniformly. Transport assistance no longer creates meaningful rate, and resource competition remains. For convenience, Table~\ref{tab:regimes} summarizes the three regimes.

\begin{table}
\centering
\caption{Fast--slow regimes}
\label{tab:regimes}
\begin{tabular}{@{}lll@{}}
\toprule
Regime & Slow transport & Allocation behavior \\
\midrule
Short deadline & $Q\approx0$ & slow rate negligible \\
Intermediate & $Q$ sensitive to $p$ & $R_f\uparrow,\ R_s\uparrow$ before $a_c$ \\
Long deadline & $Q\approx1$ & $R_f\uparrow,\ R_s\downarrow$ \\
\bottomrule
\end{tabular}
\end{table}

\subsection{Analytical illustration}
For illustration, set $\Gamma=1$ and
\begin{equation}
 p(a)=1+8a.
 \label{eq:illustrativePe}
\end{equation}
The fast channel law remains generic to ensure universal insights, but to show the allocation curve in rate coordinates, Fig.~\ref{fig:analyticcurves} uses a representative concave reparameterization
\begin{equation}
 \widetilde C_f(a)=\frac{1-e^{-3a}}{1-e^{-3}}.
 \label{eq:representativefast}
\end{equation}

At $\theta=0.08$, the slow rate is small, although assistance raises it and the maximizer is $a_c\approx0.603$. At $\theta=0.20$, the positive-slope portion is substantial and $a_c\approx0.429$. At $\theta=20$, transport is effectively saturated and the maximizing fast allocation is zero. The circles in Fig.~\ref{fig:analyticcurves} mark the slow-rate maximizers from \eqref{eq:ac}.

\begin{figure}[t]
\centering
\includegraphics[width=\columnwidth]{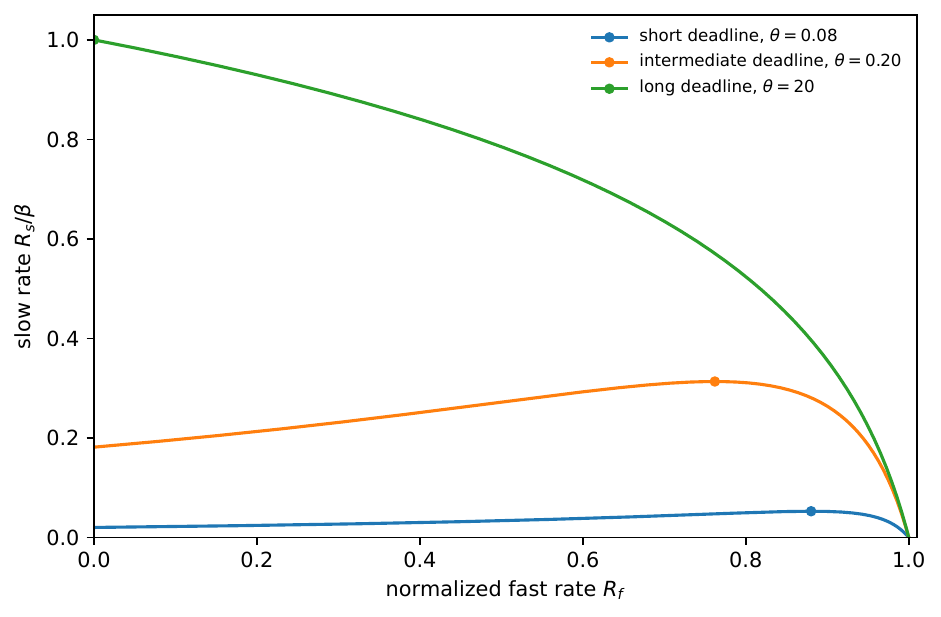}
\caption{Deterministic allocation curves for $\Gamma=1$ and $p(a)=1+8a$. Circles mark the unique slow-rate maximizer $a_c$. At short deadline the slow stream is small; at intermediate deadline the curve first rises in both rates and then falls; at long deadline ordinary competition dominates. The horizontal coordinate uses \eqref{eq:representativefast} only for visualization; the phenomenon is quite general.}
\label{fig:analyticcurves}
\end{figure}

The critical-deadline theorem is demonstrated in Fig.~\ref{fig:transition}. For these parameters,
\begin{equation}
 \theta_c\approx3.563.
 \label{eq:numericalthetac}
\end{equation}
The transition allocation decreases monotonically with the deadline and reaches zero at $\theta_c$. As $\theta\downarrow0$, it approaches $1-2/8=0.75$, in agreement with \eqref{eq:shortroot}.

\begin{figure}[t]
\centering
\includegraphics[width=\columnwidth]{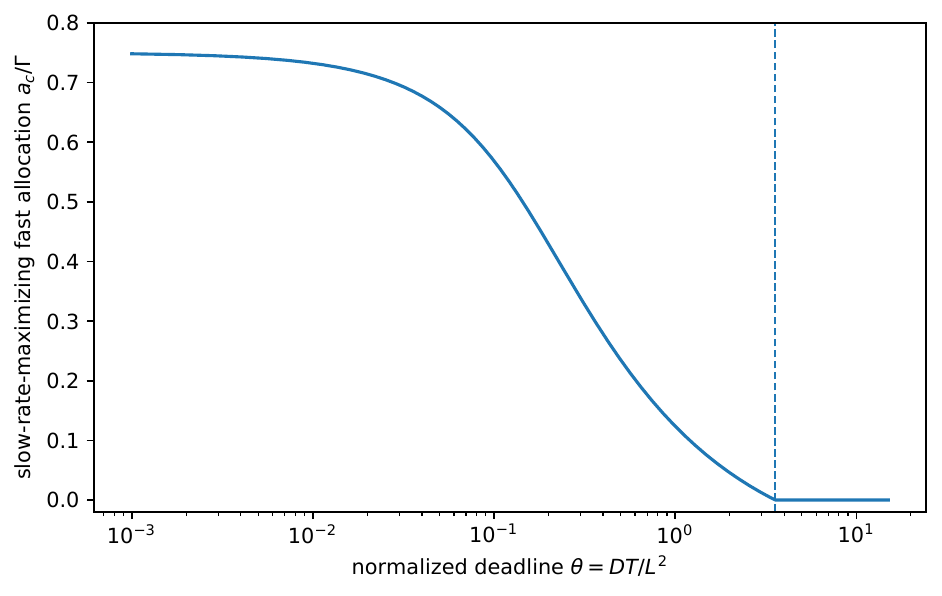}
\caption{Slow-rate-maximizing fast allocation versus normalized deadline for $\Gamma=1$ and $p(a)=1+8a$. The dashed line is the critical deadline $\theta_c\approx3.563$. More generous deadlines require less transport assistance, and beyond $\theta_c$ the slow rate is maximized at $a=0$.}
\label{fig:transition}
\end{figure}

\section{Finite-Frame LTI-Poisson Channel}
\label{sec:poisson}
The deadline-erasure model is helpful for illustrating the structural theorem, but its slow capacity is linear in molecular payload and does not contain counting noise or intersymbol interference. We therefore examine a finite-frame LTI-Poisson channel model from molecular communication \cite{aminian2015lti}.

Divide each episode into $m$ slots. For a molecule released at the beginning of a slot, define the probability of first arrival in the lag-$\ell$ slot by
\begin{equation}
 h_\ell(\theta,p)
 =Q\!\left(\frac{\ell\theta}{m},p\right)
  -Q\!\left(\frac{(\ell-1)\theta}{m},p\right),
 \quad \ell=1,\ldots,m.
 \label{eq:impulse}
\end{equation}
Let $X_j\in\{0,A\}$ be the mean number of molecules released in slot $j$. Poisson release and thinning give conditionally independent counts
\begin{equation}
 Y_k\mid X^m=x^m
 \sim\operatorname{Poisson}\!\left(
 \lambda_0+\sum_{j=1}^k h_{k-j+1}(\theta,p)x_j
 \right),
 \label{eq:poissonchannel}
\end{equation}
where $\lambda_0$ is dark current. Earlier releases contribute to later slots, so \eqref{eq:poissonchannel} contains within-frame intersymbol interference. Molecules arriving after the frame are cleared, implementing an episode reset.

Normalize the input cost as
\begin{equation}
 c(x^m)=\frac{1}{mA}\sum_{j=1}^m x_j.
 \label{eq:poissoncost}
\end{equation}
For transport state $p$ and slow budget $b$, the vector channel is a finite-input, countable-output DMC with per-slot capacity
\begin{equation}
 C_{\mathrm P}(p,b;\theta)
 =\frac1m\max_{P_{X^m}:\,\E[c(X^m)]\le b}
 I(X^m;Y^m).
 \label{eq:poissoncapacity}
\end{equation}
This can be optimized directly, see Appendix~\ref{app:numerics}.

For illustration, let $m=3$, $A=6$, $\lambda_0=0.15$, $\Gamma=0.35$, and
\begin{equation}
 p(a)=1+\frac{8}{\Gamma}a.
 \label{eq:poissonPe}
\end{equation}
The optimization in \eqref{eq:poissoncapacity} is concave in the input distribution. For numerical purposes, counts $0,\ldots,23$ are represented explicitly and all larger counts are placed in one tail bin. The largest conditional mean is at most $6.15$, for which the aggregated tail probability is below $3.9\times10^{-8}$.

Fig.~\ref{fig:poisson} shows $C_{\mathrm P}(p(a),\Gamma-a;\theta)$. Because every admissible $C_f(a)$ is increasing, an increasing portion in this plot corresponds to a positive-slope fast--slow allocation curve. The same three regimes reappear. For $\theta=0.08$, fast assistance raises capacity from $0.0044$ to $0.0761$ bits/slot, with the maximizing allocation near $0.78\Gamma$. For $\theta=0.20$, capacity rises from $0.1232$ to $0.2869$ bits/slot and peaks near $0.58\Gamma$. For $\theta=20$, capacity is maximized at $a=0$. Table~\ref{tab:poisson} gives the values.

\begin{figure}
\centering
\includegraphics[width=\columnwidth]{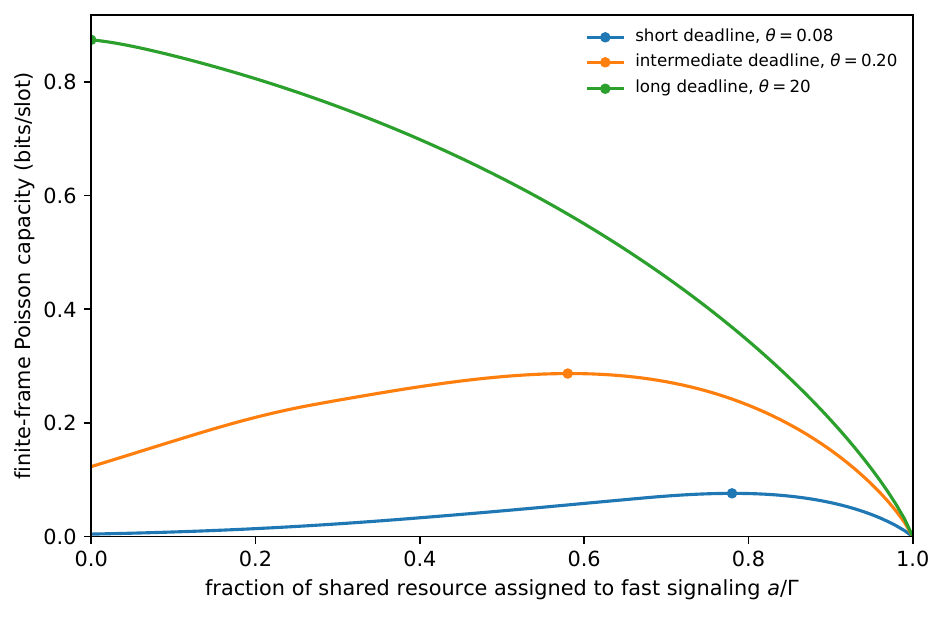}
\caption{Finite-frame LTI-Poisson capacity with counting noise and within-frame intersymbol interference. The slow capacity first rises and then falls at short and intermediate deadlines, whereas it decreases immediately at a long deadline. Circles mark the numerical maximizers.}
\label{fig:poisson}
\end{figure}

\begin{table}[t]
\centering
\caption{Finite-frame LTI-Poisson calculation}
\label{tab:poisson}
\begin{tabular}{@{}cccc@{}}
\toprule
$\theta$ & $C_{\mathrm P}(a=0)$ & $\max_a C_{\mathrm P}(a)$ & $a_{\max}/\Gamma$ \\
\midrule
$0.08$ & $0.0044$ & $0.0761$ & $0.78$ \\
$0.20$ & $0.1232$ & $0.2869$ & $0.58$ \\
$20$   & $0.8742$ & $0.8742$ & $0$ \\
\bottomrule
\end{tabular}
\end{table}

The complementarity we saw for the linear setting in \eqref{eq:erasurecapacity} is very much still present for the more complicated setting when the slow channel contains stochastic counting, a nonlinear capacity--cost function, background noise, and diffusion-induced memory within the signaling frame.

\section{Biological Interpretation}
The hydromechanical model of plant signaling mentioned in Sec.~\ref{sec:intro} motivates the separation of time scales we considered in this work. Pressure changes can move rapidly through xylem, while tissue relaxation and the associated mass flow act over longer time scales and can advect chemical elicitors to distant tissue \cite{bacheva2025hydromechanical}. In our abstraction, $a$ measures investment in the fast hydromechanical mechanism, and $p(a)$ measures the resulting quality of directed molecular transport.

The three regimes that emerged from mathematical analysis then have a natural interpretation, capturing generic structural facts about how one message-bearing process can improve the channel used by another. At very short response times, only the mechanical perturbation is plausibly available. At intermediate times, the same hydromechanical process that provides early information can enable better chemical transport. At long times, the molecular stream has arrived with high probability, so additional hydraulic assistance carries little marginal value and the two mechanisms compete for resources.

The same structure can arise beyond this specific form of plant communication. A pressure precursor can create flow for a chemical signal; an electric field can both communicate rapidly and electrophoretically assist molecular carriers; mechanical deformation can open a transport pathway; or a thermal or optical control signal can change diffusivity or reaction rates. Although field-assisted molecular communication already demonstrates the value of controlling transport \cite{ma2019electric,chou2022timevarying,cho2022electrophoretic,chou2026field}, here the assisting mechanism has its own communication role and shares a resource with the material stream.

\section{Conclusion}
Fast and slow communication do not necessarily begin in competition. When fast signaling also improves the transport law of a slower material channel, shifting resource toward the fast mechanism can initially increase both achievable rates. The sign of this effect is governed by  comparing between transport benefit and molecular opportunity cost and is universal in the sense of independent of the detailed fast-channel law.

For deadline-constrained drift diffusion, strict log-concavity of first-passage reliability yields an exact breakpoint. Under increasing concave transport actuation, the transport elasticity decreases monotonically, so there is at most one transition. An explicit initial-elasticity condition determines whether complementarity exists; a critical normalized deadline determines when it disappears; and the competitive branch is the true Pareto boundary after time sharing. A finite-frame LTI-Poisson channel exhibits the same transition in the presence of counting noise and intersymbol interference.

The resulting fast--slow communication therefore follows a sequence: at very early times the fast stream is effectively alone; at intermediate times it enables the slow stream; and at long times the two streams compete. This sequence is a property of all communication systems in which information and transport are physically coupled.

There are several possible extensions. A symbol-dependent fast waveform could actuate the full molecular impulse response, creating an action-controlled channel with memory. Shared physiological state would correlate the noise of the two streams, and the fast observation could provide side information for slow-channel state estimation. Receptor binding would add finite sensing capacity and biochemical noise. Finally, if the two streams describe the same event, the channel problem should be combined with appropriate notions of  source coding.

\bibliographystyle{IEEEtran}
\bibliography{fast_slow_endogenous_transport}

\appendices
\section{Proof of the First-Passage Reliability Theorem}
\label{app:reliability}
Let $\phin$ and $\Phin$ denote the standard normal density and cdf. Set
\begin{equation}
 c=\sqrt{\frac{\theta}{2}},
 \quad s=\Phin(A),
 \quad r=e^p\Phin(B),
 \label{eq:csr}
\end{equation}
so that $Q=s+r$. The identity
\begin{equation}
 \phin(A)=e^p\phin(B)
 \label{eq:densityidentity}
\end{equation}
follows directly from $B^2-A^2=2p$.

Differentiating \eqref{eq:Q} with respect to $p$, the two density terms cancel by \eqref{eq:densityidentity}, giving
\begin{equation}
 Q_p=r=e^p\Phin(B)>0.
 \label{eq:appQp}
\end{equation}
Consequently,
\begin{equation}
 \ell_\theta(p)=\tfrac{r}{s+r}.
 \label{eq:ellratio}
\end{equation}
Let $x=-A$ and $y=-B$. Then $y\ge x$ and $(y^2-x^2)/2=p$. The function
\begin{equation}
 g(z)=e^{z^2/2}\Phin(-z)
\end{equation}
is strictly decreasing: for $z\le0$ the derivative is plainly negative, and for $z>0$ the inequality $z\Phin(-z)<\phin(z)$ follows from
\begin{equation}
 \Phin(-z)=\int_z^\infty\phin(t)\dd t
 <\int_z^\infty\tfrac{t}{z}\phin(t)\dd t
 =\tfrac{\phin(z)}{z}.
\end{equation}
Therefore
\begin{equation}
 r=e^p\Phin(-y)\le \Phin(-x)=s,
\end{equation}
with equality only at $p=0$. Equation \eqref{eq:ellratio} proves \eqref{eq:ellbound}.

For strict log-concavity, differentiate once more:
\begin{equation}
 Q_{pp}=r-c\phin(A).
 \label{eq:Qpp}
\end{equation}
Let $\lambda(z)=\phin(z)/\Phin(z)$. Using \eqref{eq:densityidentity},
\begin{align}
 \frac{QQ_{pp}-Q_p^2}{sr}
 &=1-c\bigl[\lambda(A)+\lambda(B)\bigr].
 \label{eq:logconcavitycalc}
\end{align}
For every real $z$, $\lambda(z)>-z$. This is immediate for $z\ge0$ and, for $z<0$, follows from the same Gaussian tail bound. Since
\begin{equation}
 A+B=-\tfrac{1}{c},
\end{equation}
we obtain $c[\lambda(A)+\lambda(B)]>1$. Hence $QQ_{pp}-Q_p^2<0$, which is exactly \eqref{eq:logconcave}.

To establish monotonicity in $\theta$, let $U=D\tau/L^2$ be the
normalized first-passage time. Its density is
\begin{equation}
f_p(u)=\frac{1}{\sqrt{4\pi u^3}}
\exp\!\left[-\frac{(1-pu)^2}{4u}\right],
\qquad u>0.
\label{eq:fp-normalized-density}
\end{equation}
Thus $Q(\theta,p)=\int_0^\theta f_p(u)\,du$ and
$\partial_p f_p(u)=\tfrac12(1-pu)f_p(u)$.
Differentiation under this integral is justified because, for
$p$ in a compact interval and $0<u\le\theta$, both $f_p(u)$
and $|\partial_p f_p(u)|$ are bounded by a constant times
$u^{-3/2}e^{-1/(4u)}$.
Writing
$\mu_p(\theta)=\mathbb{E}_p[U\mid U\le\theta]$, where the
subscript denotes the drift parameter, gives
\begin{equation}
\begin{aligned}
\ell_\theta(p)
&=\frac{1}{2Q(\theta,p)}
  \int_0^\theta (1-pu)f_p(u)\,du\\
&=\frac12-\frac p2\,\mu_p(\theta).
\end{aligned}
\label{eq:fp-conditional-sensitivity}
\end{equation}
Since $f_p$ is positive on $(0,\infty)$,
$0<\mu_p(\theta)<\theta$. Differentiating the ratio defining
this conditional mean yields
\begin{equation}
\mu_p'(\theta)
=\frac{f_p(\theta)}{Q(\theta,p)}
 \bigl[\theta-\mu_p(\theta)\bigr]>0.
\label{eq:fp-truncated-mean-derivative}
\end{equation}
Consequently, for $p>0$,
\begin{equation}
\frac{\partial\ell_\theta(p)}{\partial\theta}
=-\frac{p f_p(\theta)}{2Q(\theta,p)}
 \bigl[\theta-\mu_p(\theta)\bigr]<0.
\label{eq:fp-deadline-derivative}
\end{equation}
The bound $0<\mu_p(\theta)<\theta$ also gives
$\ell_\theta(p)\to 1/2$ as $\theta\downarrow0$.
For $p>0$, the density has an exponential tail, which permits
differentiation of its normalization:
\begin{equation}
0=\int_0^\infty \partial_p f_p(u)\,du
=\frac12-\frac p2\,\mathbb{E}_p[U].
\label{eq:fp-unconditional-mean}
\end{equation}
Hence $\mathbb{E}_p[U]=1/p$. As $\theta\to\infty$,
monotone convergence of the truncated first moment, together
with $Q(\theta,p)\to1$, gives $\mu_p(\theta)\to1/p$.
Equation~\eqref{eq:fp-conditional-sensitivity} therefore gives
$\ell_\theta(p)\to0$. At $p=0$, the same identity instead gives
$\ell_\theta(0)=1/2$ for every finite $\theta$.

\section{Proof of the Exact Allocation Theorem}
\label{app:allocation}
Differentiate \eqref{eq:elasticity}:
\begin{align}
 \Ecal_\theta'(a)
 &=-p'(a)\ell_\theta(p(a))
 +(\Gamma-a)p''(a)\ell_\theta(p(a))\nonumber\\
 &\quad +(\Gamma-a)[p'(a)]^2\ell_\theta'(p(a)).
 \label{eq:Eprime}
\end{align}
The first term is strictly negative; the second is nonpositive by concavity of $p$; and the third is strictly negative by strict log-concavity of $Q$. Thus $\Ecal_\theta$ is strictly decreasing. Moreover, $\Ecal_\theta(a)\to0$ as $a\uparrow\Gamma$. Equation \eqref{eq:Sprimeelasticity} therefore implies that a positive-derivative interval exists exactly when $\Ecal_\theta(0)>1$, and then there is a unique crossing $a_c$.

For concavity on the decreasing branch, abbreviate $Q=Q(\theta,p(a))$ and $b=\Gamma-a$. Twice differentiating \eqref{eq:Stheta} gives
\begin{equation}
 \tfrac{S_\theta''(a)}{\beta}
 =-2Q_p p'+bQ_p p''+bQ_{pp}(p')^2.
 \label{eq:Ssecond}
\end{equation}
Strict log-concavity implies
\begin{equation}
 \tfrac{Q_{pp}}{Q_p}<\tfrac{Q_p}{Q}=\ell_\theta.
 \label{eq:Qppbound}
\end{equation}
When $a\ge a_c$, $\Ecal_\theta(a)=bp'\ell_\theta\le1$. Hence
\begin{align}
 \tfrac{S_\theta''(a)}{\beta}
 &< Q_p p'[-2+b p'\ell_\theta]+bQ_p p''\nonumber\\
 &= Q_p p'[-2+\Ecal_\theta]+bQ_p p''\nonumber\\
 &\le -Q_p p'<0,
 \label{eq:Sconcave}
\end{align}
where the last inequality uses $\Ecal_\theta\le1$ and $p''\le0$.
Thus $S_\theta$ is strictly concave and decreasing on $[a_c,\Gamma]$.

Let $G(r)=S_\theta(F^{-1}(r))$. Where derivatives exist,
\begin{equation}
 G''(r)=\frac{S_\theta''F'-S_\theta'F''}{(F')^3}.
 \label{eq:Gsecond}
\end{equation}
On the decreasing branch, $S_\theta''<0$, $S_\theta'\le0$, $F'>0$, and $F''\le0$, so $G''<0$. The hypograph below this branch is therefore convex. Points with $a<a_c$ are dominated by $a_c$, and time sharing cannot lift the decreasing branch. It is the Pareto boundary of \eqref{eq:capacityregion}. Nondifferentiable concave $C_f$ may be handled by one-sided derivatives or subgradients.

\section{Numerical Capacity Evaluation}
\label{app:numerics}
For the LTI-Poisson calculation, the $2^m$ binary release vectors are the input alphabet. Given an input vector, \eqref{eq:poissonchannel} yields a product distribution over slot counts. After tail aggregation, this is a finite DMC with transition matrix $W$. For each $a$, we solve
\begin{equation}
 \max_{P_X}\ I(P_X,W)
 \quad\text{subject to}\quad
 \sum_x P_X(x)c(x)\le\Gamma-a.
 \label{eq:dmcoptimization}
\end{equation}
Mutual information is concave in $P_X$, so every optimum of \eqref{eq:dmcoptimization} is global. The reported values use a constrained convex optimization over the probability simplex.
\end{document}